\documentclass{article}
\usepackage{makeidx}  
\usepackage[margin=1in]{geometry}
\usepackage{graphicx}
\graphicspath{./fig/}
\usepackage{amsmath}
\usepackage{hyperref}
\usepackage{tikz}
\usepackage{algorithm,amssymb}
\usepackage{algpseudocode} 
\usepackage[english]{babel}
\usepackage{mathtools}
\usepackage{multirow}
\usepackage{authblk}
\usepackage{amsthm}
\usetikzlibrary{calc}
\newcounter{cumulArea}
\newtheorem{theorem}{Theorem}

\newcommand{\Expect}{{\mathbb{E}}}
\newcommand{\Region}{{\mathcal{R}}}
\newcommand{\B}{{\mathcal{B}}}
\newcommand{\G}{{\mathcal{G}}}
\newcommand{\E}{{\mathcal{E}}}

\newcommand{\floor}[1]{\lfloor #1 \rfloor}

\begin{document}
%

\title{Lookup Lateration: Mapping of Received Signal Strength to Position for Geo-Localization in Outdoor Urban Areas}



\author[1]{Andrey Shestakov}
\author[2]{Danila Doroshin}
\author[2]{Dmitri Shmelkin}
\author[1]{Attila Kert\'esz-Farkas\thanks{Correspondence to akerteszfarkas@hse.ru}}
\affil[1]{Department of Data Analysis and Artificial Intelligence, Faculty of Computer Science, National Research University Higher School of Economics (HSE), Moscow, Russia}
\affil[2]{Huawei Technologies Co. Ltd, Moscow, Russia}

%

\maketitle              

\begin{abstract}
The accurate geo-localization of mobile devices based upon received signal strength (RSS) in an urban area is hindered by obstacles in the signal propagation path. Current localization methods have their own advantages and drawbacks. Triangular lateration (TL) is fast and scalable but employs a monotone RSS-to-distance transformation that unfortunately assumes mobile devices are on the line of sight. Radio frequency fingerprinting (RFP) methods employ a reference database, which ensures accurate localization but unfortunately hinders scalability. 

Here, we propose a new, simple, and robust method called lookup lateration (LL), which incorporates the advantages of TL and RFP without their drawbacks. Like RFP, LL employs a dataset of reference locations but stores them in separate lookup tables with respect to RSS and antenna towers. A query observation is localized by identifying common locations in only associating lookup tables. Due to this decentralization, LL is two orders of magnitude faster than RFP, making it particularly scalable for large cities. Moreover, we show that analytically and experimentally, LL achieves higher localization accuracy than RFP as well. For instance, using grid size 20 m, LL achieves 9.11 m and 55.66 m, while RFP achieves 72.50 m and 242.19 m localization errors at 67\% and 95\%, respectively, on the Urban Hannover Scenario dataset. 
\end{abstract}

\section{Introduction}
The localization of wireless mobile devices \cite{Zekavat2011,Gezici2008} has become a key issue in emergency cases \cite{FCC1996,EC2002}, surveillance, security, family tracking, etc. GPS-based localization, while very accurate, is usually not preferred because of its high energy consumption; therefore, a  wide range of technologies has emerged for localization based on received signal strength (RSS), time of arrival (ToA), and angle of arrival (AoA). The two latter methods, ToA and AoA, require additional equipments, such as an antenna array to time synchronization between the transmitter and the receiver, while AoA requires some array to identify the signal's angle. However, RSS measurements are ubiquitous and readily available in almost all wireless communication systems; hence, it seems plausible to gain information about mobile devices' position.

Our main problem is characterized as follows: We assume that we are given a set of antennas $A=\{a_i=(a_{i_x},a_{i_y}) \in R^2\}$ and several RSS measurements for a single mobile phone: 
\begin{equation}
M_{RSS}=[r_{a_1}, r_{a_2}, ...., r_{a_{|A|}}],
\end{equation}
\noindent where $r_{a_i}$ denotes RSS from antenna $a_i$, measured in dBm  if it is observed; otherwise, it is defined as $NaN$. Antenna $a_i$ and its location will be used interchangeably in this article. The main task is to determine the mobile phone's location, where its RSS measurements were observed. A straightforward method to carry out a simple localization is the triangular lateration (TL). This method calculates the mobile phone's location by solving the following optimization problem: 
\begin{equation}
\begin{aligned}
\underset{(x,y)}{\text{min}}\;\quad&\sum_{\mathclap{i:r_{a_i} \in M_{RSS}} }\left(  (a_{i_x}-x)^2 + (a_{i_y}-y)^2 - (d(r_{a_i}))^2 \right)^2,
\label{eq:lateration}
\end{aligned}
\end{equation}
where $d(.)$ denotes a distance function. When the distance function $d(.)$ is monotone and at least three measurements are available, the optimization problem becomes convex and avoids the local minima phenomenon. Although this method is widely used because of its simplicity, there are two main issues with it. First, it requires a well-calibrated signal-to-distance conversion function $d(.)$. For urban areas, the COST-HATA models were developed by the COST European Union Forum based on various field experiments and research \cite{neskovic2000modern}.
Second, it assumes mobile devices on line of sight (LoS). This requirement is hardly met in practice in urban areas because signal propagation is hindered by obstacles, concrete constructions, buildings, churches, tunnels, underpasses, etc. In this case, the observed distance can be decomposed as
\begin{equation}
d(r_{a_i})=L_i+e_i+NLOS_i,
\end{equation}
where $L_i$ is the true distance and $e_i$ is the receiver noise, which is assumed to be a zero mean Gaussian random variable. The quantity $NLOS_i$ denotes the error caused by obstacles, buildings, and constructions on the signal propagation path. It is worth mentioning that according to field test results, $NLOS_i \gg e_i$. These obstacles interfere with the signal by either weakening it or changing its path. The latter phenomenon is called multipath propagation. In either case, the antenna receives weaker signals, which indicates that mobile phones are located farther away than they are. This is illustrated in Figure \ref{fig:3Lateration}A. The NLOS error can be modeled by a Gaussian or exponentially distributed error  \cite{cong2005}; however, it can be considered a deterministic error that depends only on obstacles located between antenna and mobile phones. An article \cite{Ma2007} proposed an error mitigation method based on the distribution of the circle of positions' intersections. LMedS algorithm \cite{Marco2008} is based on the observation that the intersection generated by an obstacle on the sight is far from the other intersections. Since the method in eq. \ref{eq:lateration} is sensitive to outliers because of the squared error measure $(l_2)$, LMedS uses an outlier insensitive optimization to solve that minimization.  These methods show promising results when few obstacles can be found on the landscape, though how these methods will work in a dense urban area remains to be seen. 

\begin{figure}[tbp]
	\centering
	\small
	\begin{tabular}{cc}
	\includegraphics[trim=3cm 2cm 3cm 2cm,clip=true,width=70mm]{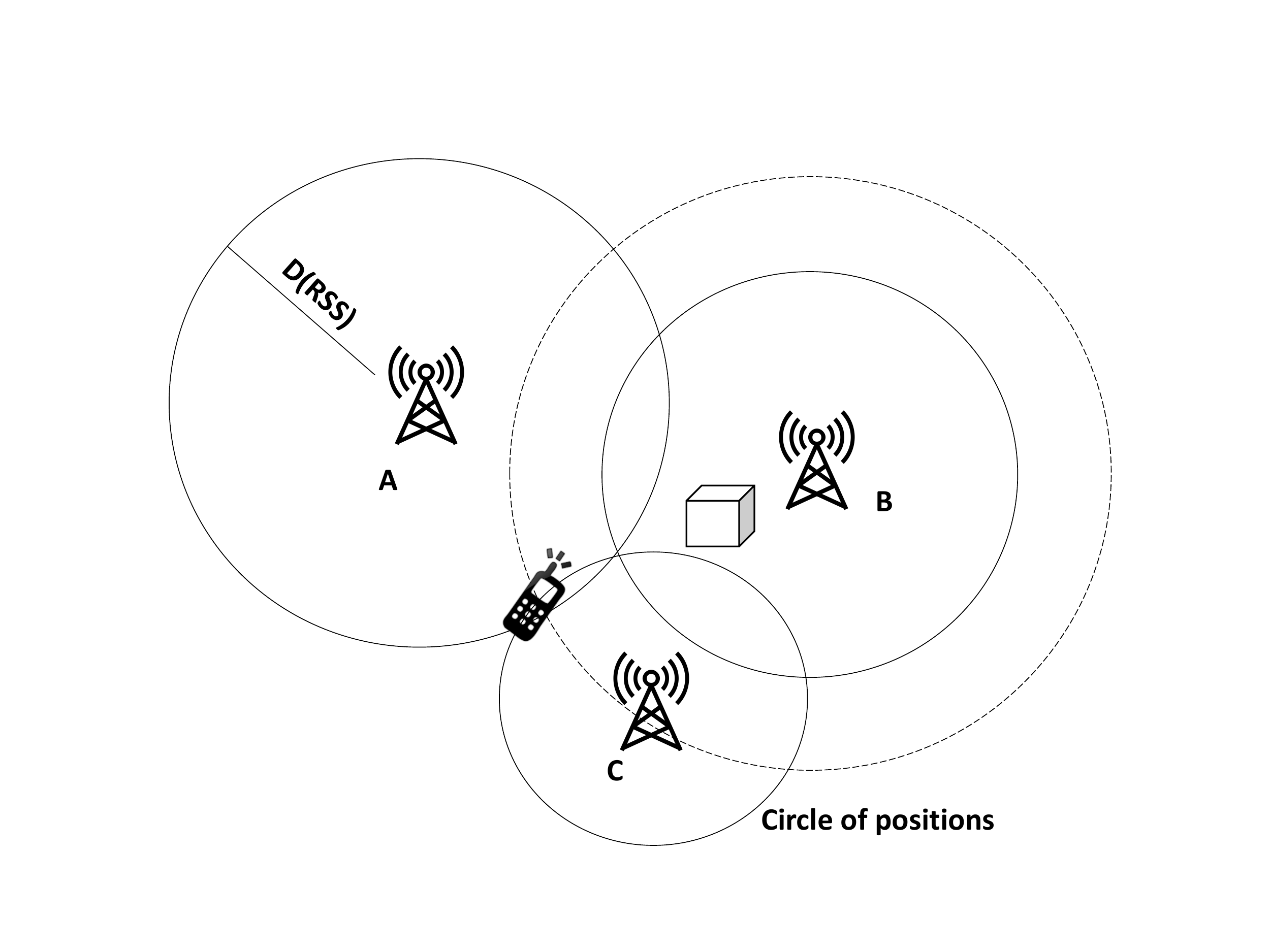}&	\includegraphics[width=70mm]{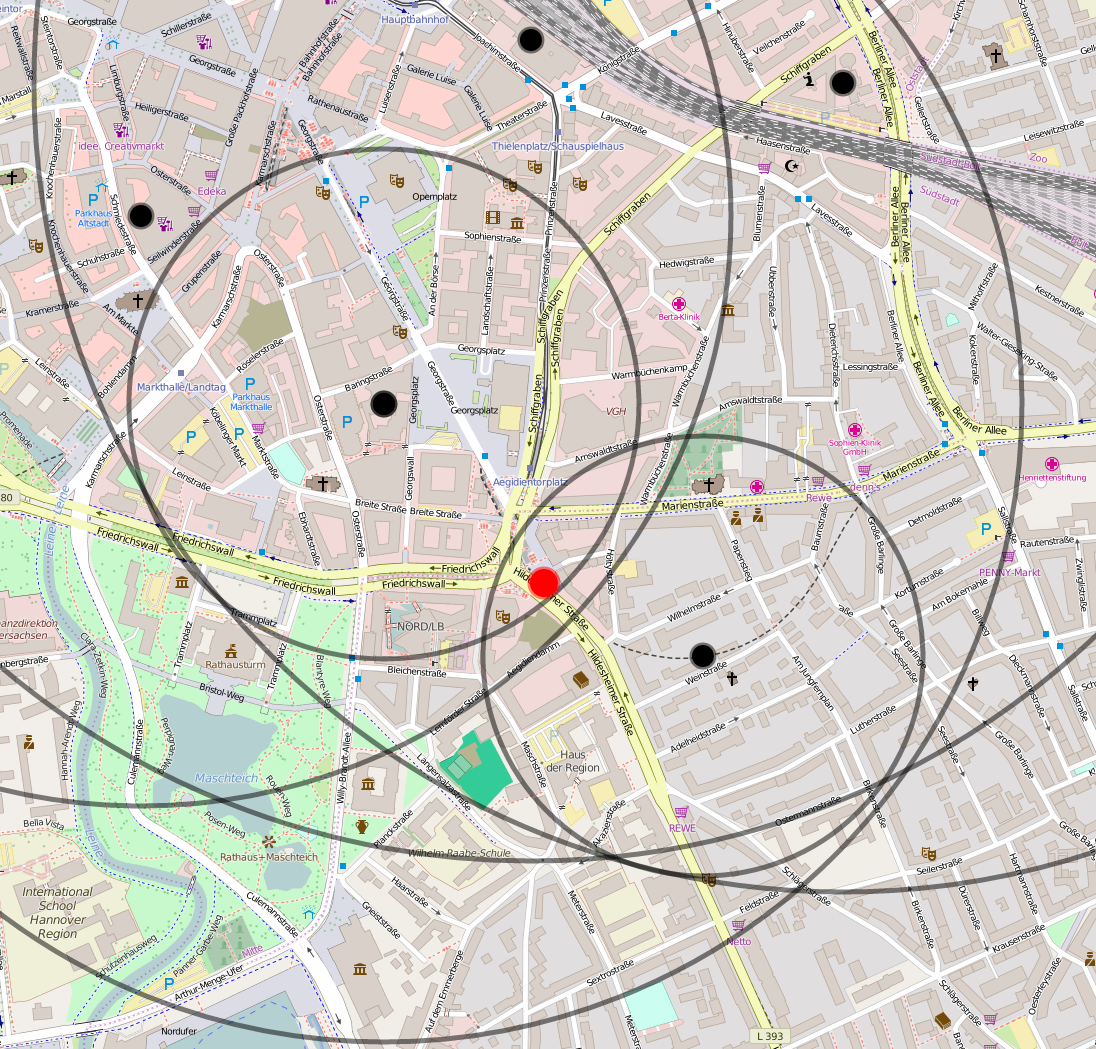}\\
	(A) &	(B) \\
	\end{tabular}	
	\caption{Lateration. A) Triangular lateration hindered by obstacles. D(RSS) denotes the distance determined by RSS. The dashed circle denotes the true distance of a mobile phone from tower B, which appears farther because of the obstacle (represented by a box) between them. B) Lateration in Hannover dataset using the six strongest RSSs. The red dot denotes the mobile device's true location, the black dots denote antennas, and the circles denote the circle of positions obtained with the COST-HATA model based on actual RSSs. Data (User ID=10048 and Time=0.1) was taken from article \cite{Rose2013}.}
	\label{fig:3Lateration}
\end{figure}

To characterize the NLoS error, we analyzed the Urban Hannover Scenario dataset \cite{Rose2013}, which contains approximately 22 Gb of localized RSS measurements from Hannover. For further details, we refer the reader to Results section or article \cite{Rose2013}.

To investigate the relationship between observed RSSs and true distances, we first selected tower ID=183 and identified mobile phones located at the fourth degree from the tower's azimuth. The data are shown in Figure \ref{fig:RSS-vs-dist}A, where the blue triangle denotes the tower's location and where the blue dots indicate the mobile phones' locations. The selection of the tower and the degree was arbitrary. Then we plotted the measured RSSs against the true distance shown in Figure \ref{fig:RSS-vs-dist}B. The relationship indicates an exponential-like correlation. However, the weaker the signal is, the noisier the correlation becomes. Signals stronger than -60 dBm could be converted to distance unambiguously. But as signal strength decreases, the conversion becomes more ambiguous. For instance, an observed -90 dBm signal could indicate a mobile phone located somewhere between 2.4 km and 4.5 km from the given tower. It is also worth noticing that the phones are not distributed evenly in this range; rather, they are concentrated at certain distances. For instance, no phone is located between 2.5 km and 3.0 km in this direction. This gap can be seen at the railways shown in Figure \ref{fig:RSS-vs-dist}A.

\begin{figure}[btp]
	\centering
	\small
	\begin{tabular}{cc}
		\includegraphics[width=70mm]{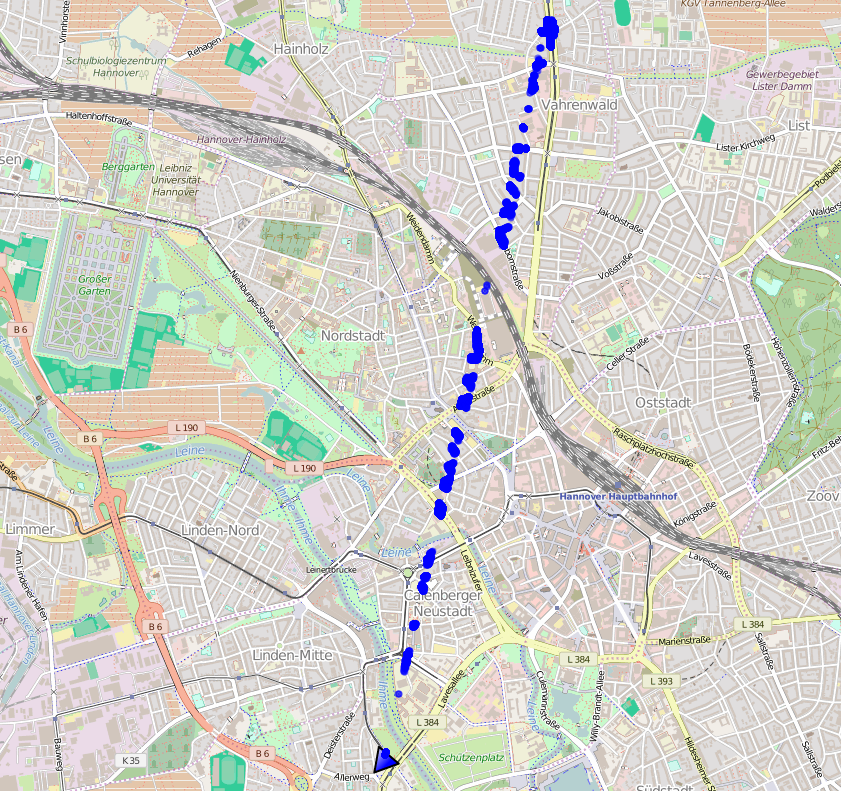}&
		\includegraphics[trim=1.0cm 6cm 2cm 7cm,clip=true,width=70mm, height=50mm]{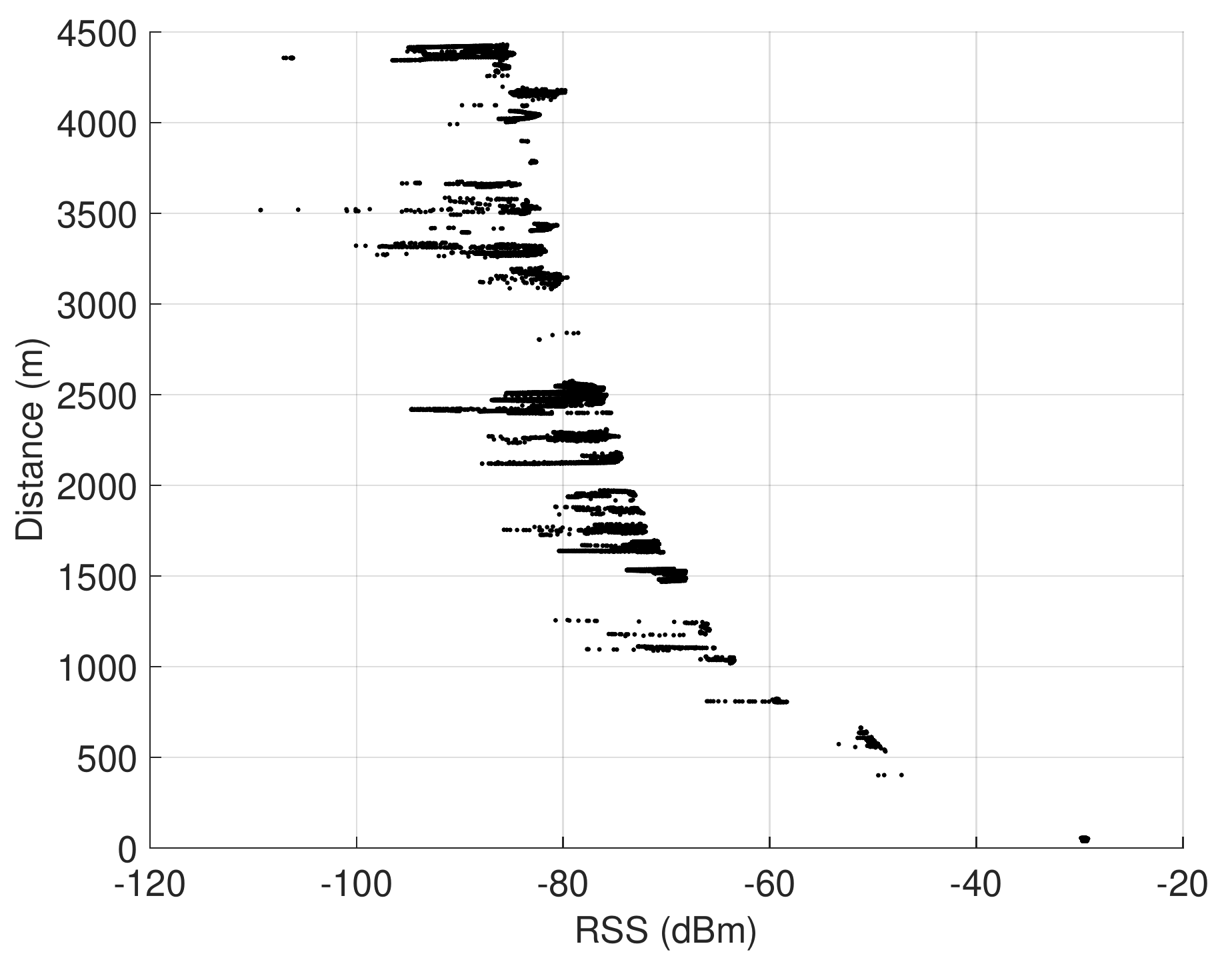}\\
		(A)&
	    (B) \\
	\end{tabular}	
	\caption{RSS vs. distance. A) Location of mobiles phones (marked blue dots) from antenna ID=183 (marked blue triangle at the bottom) at 4th degree to its azimuth (clockwise). B) Correlation between the distance measured from antenna and the RSS. }
	\label{fig:RSS-vs-dist}
\end{figure}

Next, we investigated the circle of positions corresponding to a given RSS at a tower. On a map, we placed the locations of the mobile phones that receive -79 dBm from the tower ID=183 from any direction. This map is shown in Figure \ref{fig:RSS-dist-loss}A. We can observe that the locations lay over on a belt 2 km wide rather than on a circular line. To show that this is not a measurement noise but a result of the presence of obstacles,
we plotted mobile phones receiving -78, -79, and -80 dBm, respectively, in Figure \ref{fig:RSS-dist-loss}B-C. On this map, mobile phones receiving signals of the same RSS group together and do not mix randomly. Hence, we conclude that the signal noise is small and that the error is driven by obstacles.

The radio frequency fingerprinting (RFP) \cite{Bahl2000,Gentile2012} approach is designed to overcome the NLoS problem. This is a two-phase algorithm. The first step, called offline training phase or surveying, is to construct a reference dataset that consists of a collection of localized RSS measurement vectors $M_{XY}=[r_{a_1}, r_{a_2}, \dots, r_{a_{|A|}}]$ all over the area, where $x,y$ denotes its true position. In the second phase, localization phase, the query device's position producing $M_{RSS}$ is estimated based on the nearest dataset member via a k-Nearest Neighbor ($k$NN) approach. In other words, the location $x,y$ is determined by solving the following search problem:

\begin{equation}
\begin{aligned}
&\tilde{x},\tilde{y}=  &\underset{x,y}{\text{argmin}}&\; d(M_{RSS}, M_{XY}), \\
\end{aligned}
\end{equation}
where the distance function $d(.,.)$ can be, for instance, a Euclidean distance, a cosine distance \cite{He2016}, or Jaccard coefficients \cite{Jiang2012}, which simply omit $NaN$ values. 
One drawbacks of RFP is the lack of scalability, which means that it can be considerably slow over large areas.
\begin{figure}[btp]
	\centering
	\small
	\begin{tabular}{ccc}
		\includegraphics[width=40mm, height=50mm]{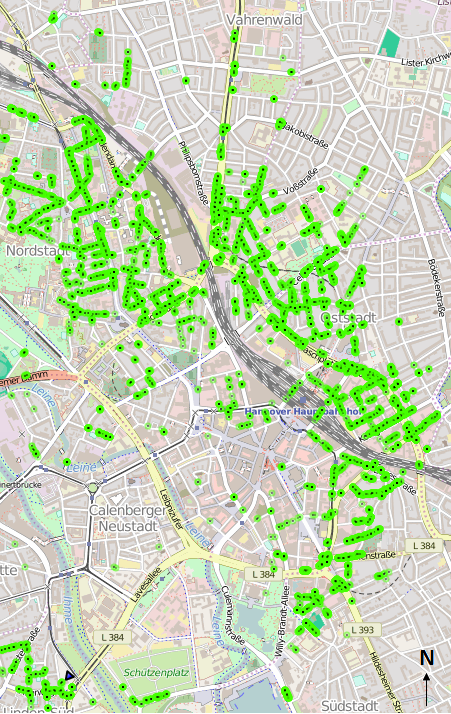} &
		\includegraphics[width=40mm, height=50mm]{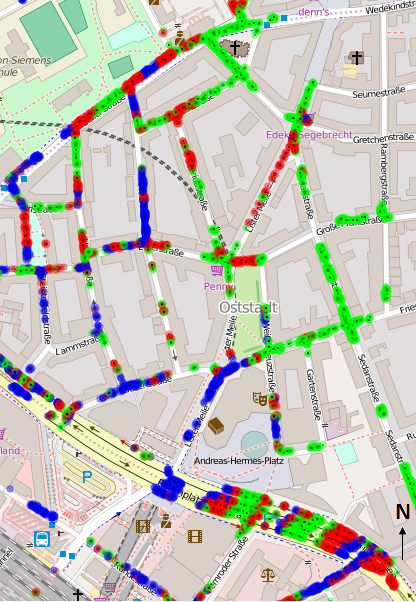} &
		\includegraphics[width=40mm, height=50mm]{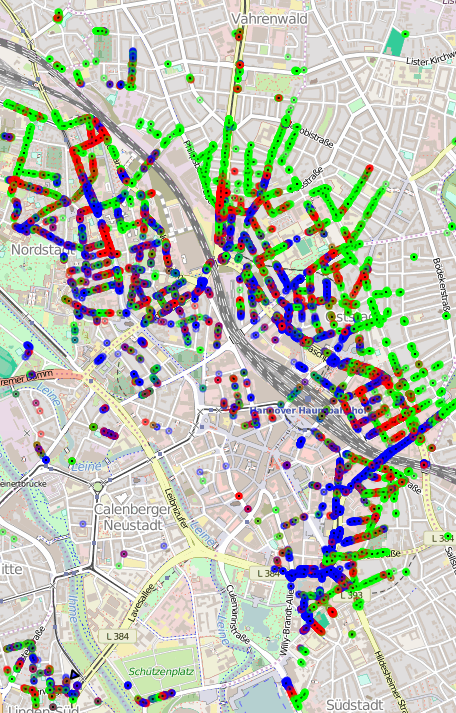}\\
		(A) & (B) &(C) \\
	\end{tabular}	
	\caption{Mobile phone locations from antenna ID=183. A) Mobile phones receiving -79 dBm signals. B) Mobile phones receiving -78 dBm (blue), -79 dBm (red), and -80 dBm (green) RSSs nearby Hauptbahnhof (central train station). C) Same as B) over downtown.}
	\label{fig:RSS-dist-loss}
\end{figure}
Grid systems are often utilized to mitigate data redundancy, making localization faster with less memory consumption. 
Measurements within a given grid can be replaced with their mean vectors in the reference dataset, and mobile phones can then be localized in the center of the closest grid. Therefore, RFP can be formulated as follows: localize query $M_{RSS}$  in the center of grid $\tilde{xy}$ calculated by $\tilde{xy} = \text{argmin}_{xy\in \G}\; d(M_{RSS}, \mu_{xy})$, where $\mu_{xy}$ is the sample mean observed in grid $xy$, $\G$ denotes the set of grids and $d(.,.)$ is a suitable distance function. This reduces search time and the size of the reference datasets by a couple of magnitudes at the expense of localization accuracy, and the trade-off can be controlled with the grid's size. To make localization even faster, Campos et al. \cite{CamposL09} have developed filtering procedures to reduce the search space.   

Besides RFP, which is based on a $k$-NN, several classic machine learning methods have been evaluated and tested on localization problems. For instance, Wu et al. \cite{wu2004wlan} have tried support vector machines, Barrau et al. \cite{nuno2006new} have used linear discriminant analysis, and Campus et al. and Magro et al. \cite{campos2010mobile,magro2007genetic} have used genetic algorithms for positioning measurements. Artificial neural networks are also popular; they can learn a regression function to map a $M_{RSS}$ to its corresponding location $\tilde{x},\tilde{y}$ \cite{Yu2011,Takenga2006,Takenga2006a}. No matter which methods are used, they all need to learn highly nonlinear mapping like in Figure \ref{fig:RSS-vs-dist}B to provide a bona fide approximation. In our opinion, this seems quite challenging in practice.

In this paper, we present a new method called lookup lateration (LL) for the rapid and accurate geo-localization of mobile phones based on RSS. The underlying idea is based on the decentralization of the reference dataset. For every antenna tower $a_i\in A$, LL builds a lookup table to store associating reference locations, that is, $\tau_{a_i}(r_{a_i})=\{(XY)\mid  r_{a_i}\in M_{XY}\}$. Therefore, the localization of a given query $M_{RSS}$ can be carried out by simply determining the common locations in the associating lookup tables, which can be formally stated as follows:
\begin{equation}
\tilde{x},\tilde{y}= \bigcap_{r_{a_i} \in M_{RSS} }\tau_{a_i}(r_{a_i}).
\end{equation}
One of the advantages of our method compared to RFP is that LL does not need to search the whole reference dataset, only two-six tables associated to query measurements. This results in a great acceleration; LL is around 100 times faster than RFP on the Urban Hannover Scenario dataset, and we believe LL can be even faster in very big cities. 

LL resembles TL to some extent. Every lookup table can be considered an RSS-to-distance nonlinear mapping $d(.,.)$ (more precisely, a relation) w.r.t. a given antenna tower. However, instead of solving a nonconvex optimization problem, LL carries out the localization by determining common elements in the lookup tables. Therefore, LL does not involve local minima problems, but it may result in multiple locations, which need to be addressed.   

A case where different RSS measurements from the same reference grid are observed from the same tower is worth commenting on. In this case, RFP would store only the average of these measurements, which would result in information loss. However, LL would store the reference grids in multiple lookup tables with respect to the measured RSS. Thus, LL preserves all the information, and this offers a great advantage in terms of localization accuracy. 

Our contribution is summarized as follows. In the second section, we formally introduce the lookup lateration (LL) algorithm. In the third section, we give a formal comparison on the error obtained with LL and RFP, and we point out a conceptual limitation of RFP when it is used with grid systems. In the fourth section, we present and discuss our experimental results. Finally, we summarize our findings in the last chapter. 

\section {Lateration Using RSS-to-Location Lookup Tables}

In the previous section, we concluded that the relationship between received RSS and true distance is nonmonotone in dense urban areas. Using any nonmonotone, continuous function $d(.)$ in eq \ref{eq:lateration} would result in a nonconvex optimization problem. Here, we introduce a new procedure that we termed lookup lateration (LL).

The procedure is provided a collection of $M_{XY}$ measurements annotated with their true locations as training data. Then lookup tables $\tau_a(r)$ are constructed, which contain mobile locations with respect to RSS $r$ measured by antenna $a$. Because RSS measurements $r$ are real valued and hindered by some measurement noise, we simply applied binning techniques to group measurements together. The corresponding bin $b$ of a measurement $r$ using bin size $s$ is calculated $b=\floor{r/s}$. In most of our experiments, we simply used bin $s=1$; thus, all $r$ were rounded down to the closest integer. For instance, -64.7 dBm is rounded down to -64. In the rest of the paper, we assume that for LL, all measurements $r$ are already binned with  $s=1$, unless it is specified otherwise. 

The lookup table $\tau_a(r)$ implicitly encodes distances corresponding to observed RSS measurements. This procedure is shown by Algorithm \ref{alg:creation-lookup-table}. Note that $\tau_a(r)$ could contain explicit distances between the antenna $a$ and the locations of mobile phones; however, as seen in the next step, this is not necessary. Figure \ref{fig:RSS-dist-loss}A shows an example, where data were taken from the Urban Hannover Scenario. Green dots mark the locations stored in $\tau_{183}(-79)$. To avoid redundant locations in lookup tables and reduce their size, nearby locations can be grouped using clustering algorithms, and the center of clusters can be stored in lookup tables.

\begin{algorithm}
	\begin{algorithmic}[1]
		\Procedure{ConstructLookupTables}{M,D}
		\For {all antenna tower $a$}
			\For { all RSS $r$}
				\State $S \leftarrow \{(x_i,y_i)\mid r_{a_i}=r, a_i=a\}$
				\State $\tau_a(r) \leftarrow$ Clustering($S$, D)
				\Comment Group mobiles nearby.
			\EndFor
		\EndFor
		\State \Return $\tau$
		\EndProcedure
	\end{algorithmic}
	\caption{{\bf Construction of lookup tables.} The input is a list of localized RSS measurement $M=\{(x_i,y_i, r_{a_i})\}_{i=1}^N$. Triplet $(x_i,y_i, r_{a_i})$ denotes location $x_i,y_i$, where RSS $r_{a_i}$ is measured from tower $a_i$. $D$ is the maximal allowed diameter of a cluster.
		\label{alg:creation-lookup-table}}
\end{algorithm}

Now the next step is to determine the location of a given measurement $M_{RSS}=[r_{a_1}, r_{a_2}, ...., r_{a_{|A|}}]$. Here, the principle is that it can be carried out by determining the common locations in the corresponding tables $\tau_{a_i}(r_{a_i})$, that is, $M_{RSS}$ is annotated by the location $\bigcap_{i}\tau_{a_i}(r_{a_i})$ for $r_{a_i}\ne NaN$ and $\tau_{a_i}(r_{a_i})\ne\emptyset$. In practice, this could lead easily to an empty set or unambiguous locations.  Our method, shown in Algorithm \ref{alg:lateration-lookup-table}, is slightly different as, in a greedy manner, it takes into account that stronger signals provide more reliable information. Let $M=\{r_{a_{i_1}}, r_{a_{i_2}}, ...., r_{a_{i_n}}\}$ be a list of the observed RSS  from $M_{RSS}$ in decreasing order. Our algorithm starts with the set of candidate locations $C^1=\tau_{a_1}(r_{a_1})$ provided by the strongest signal. In subsequent iterations, in the while loop at line \ref{alg:term}, candidate location $c_i\in C^1$ is eliminated from $C^1$ if it does not appear as a candidate location from another antenna $a_k$ within a tolerance $T$. This iteration terminates if either all measurements are processed or the  candidate locations in $C^k$ have a smaller variance than a predefined threshold. The iteration also terminates when $C^k$ is emptied. Figure \ref{fig:lookup_lateration} shows an example of how this algorithm works.

\begin{algorithm}
	\begin{algorithmic}[1]
		\Procedure{LookupLateration}{$M_{RSS}$, $\tau$, $T$}
		\State Remove $ r_{a_i}\in M_{RSS}$ from $M_{RSS}$ {\bf if}  $\tau_{a_i}(r_{a_i})$ is empty
		\State \Return some default location {\bf if} $M=\emptyset$ \Comment e.g., location of antenna tower $a_1$.
		\State $k = 1$
		\State $C^k \leftarrow \tau_{a_k}(r_{a_k})$ 
		\While {$C^k$ is unambiguous and $k\le |M_{RSS}|$} \label{alg:term}
			\State $k\leftarrow k+1$
			\State $C^{k} \leftarrow \{c_i | c_i\! \in\! C^{k-1},c_j\!\in\! \tau_{a_k}(r_{a_k}),d(c_i,c_j)\! \le\! T\}$ \label{alg:dist}
		\EndWhile
		\State $C^k \leftarrow C^{k-1}$ {\bf if} $C^k$ is empty
		\State \Return mean of locations in $C^k$ for query $M_{RSS}$
		\EndProcedure
	\end{algorithmic}
	\caption{{\bf Lookup lateration.} The input $M_{RSS}=\{r_{a_1}, r_{a_2}, ...., r_{a_n}\}$ is a list of RSS measurements in decreasing order. Returns a location estimation for $M_{RSS}$.
		\label{alg:lateration-lookup-table}}
\end{algorithm}

\begin{figure}[hbtp]
	\centering
	\small
	\begin{tabular}{ccc}
		\includegraphics[width=50mm,]{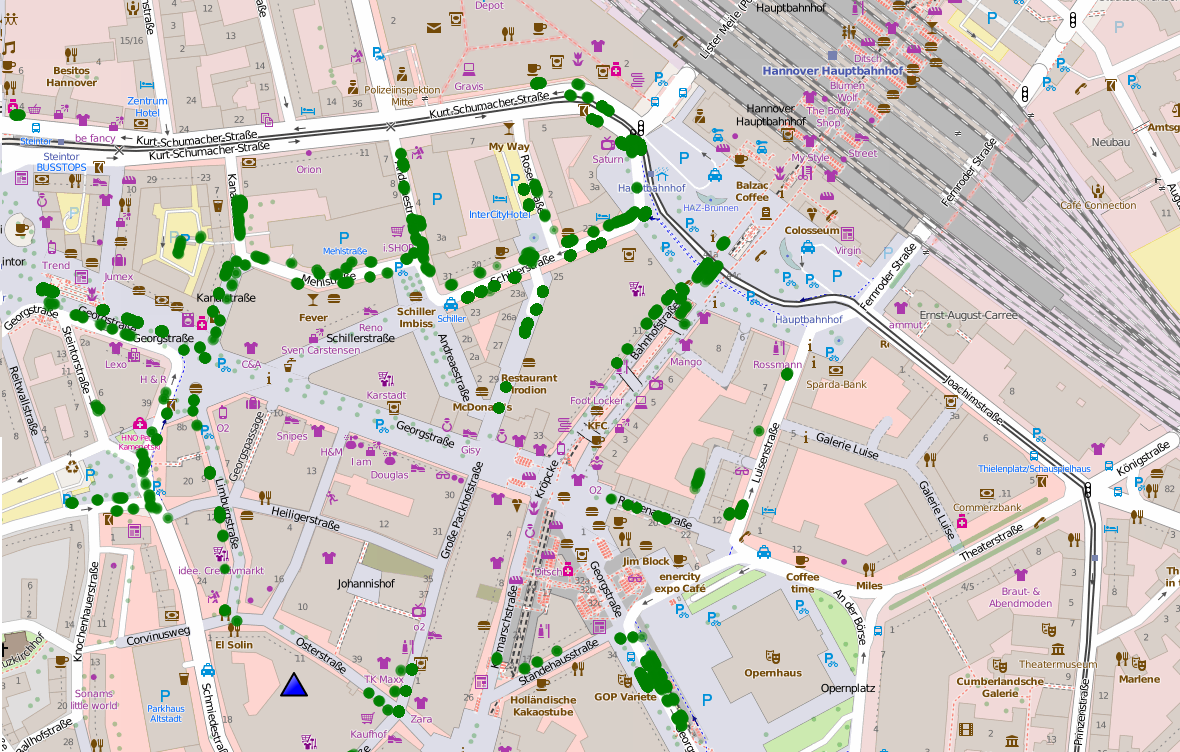} &
		\includegraphics[width=50mm,]{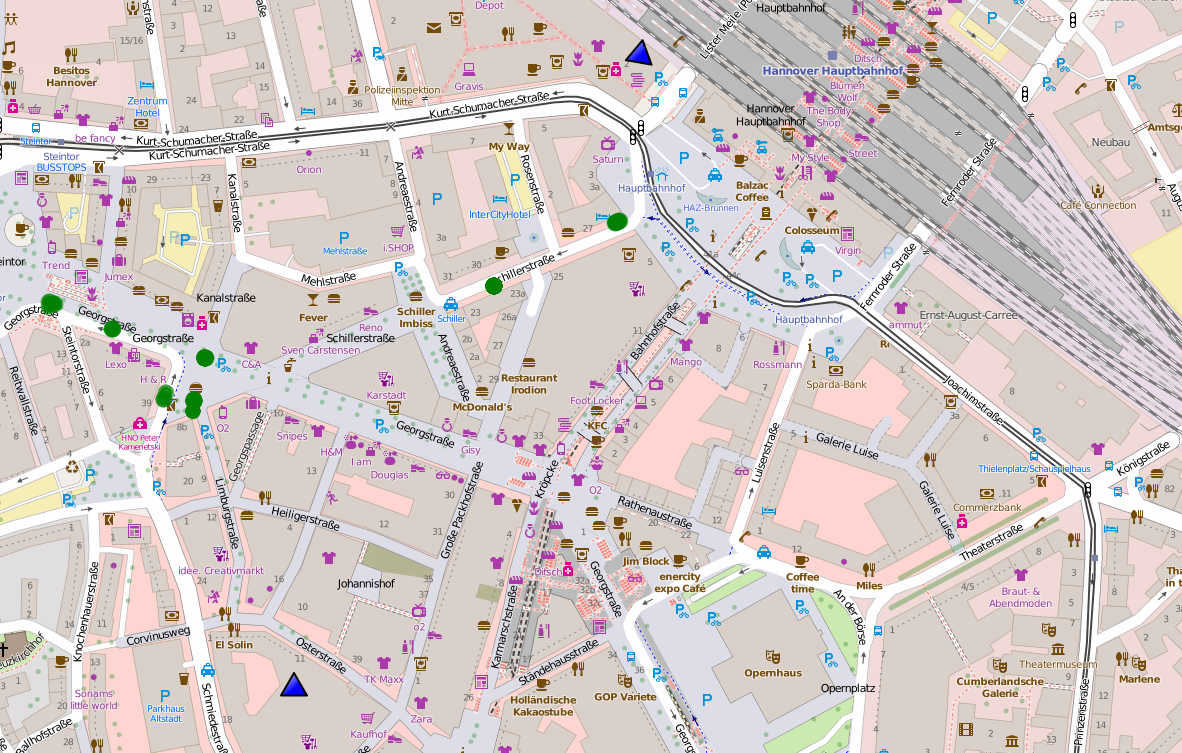} &
		\includegraphics[width=50mm, ]{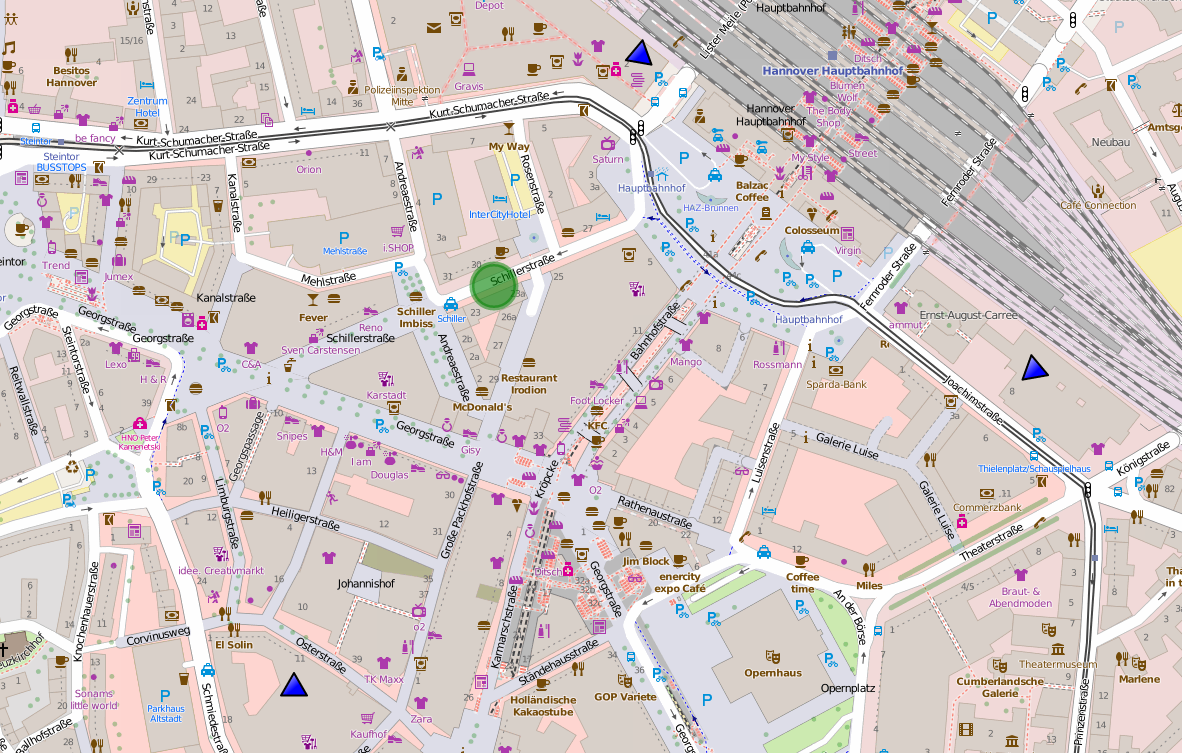} \\
		(A) &
		(B) &
		(C)\\
	\end{tabular}	
	\caption{Illustration of lookup lateration for query $M_{RSS}=\{88\!\!:\!\!-50;\; 27\!\!:\!\!-55;\; 135\!\!:\!\!-55\}$. Candidate locations are marked by green dots, and three antennas (ID=88 [bottom], ID=27 [top] ,ID=135 [right hand side]) are marked by blue triangles. A) Initial locations in  $C^1=\tau_{88}(-50)$ at the beginning of Algorithm \ref{alg:lateration-lookup-table}.  B) Next, locations not present in $\tau_{27}(-55)$ are removed, and the resulting locations in $C^2$ are shown. C) Points not present in $\tau_{135}(-55)$ are also removed, yielding an ambiguous location estimation in $C^3$ for $M_{RSS}$. }
	\label{fig:lookup_lateration}
\end{figure}

\section{Error Estimation in Grid Systems}

The LL method can also be used with grid systems, which can also lead to further simplifications of the algorithm. In our experiments, Algorithm \ref{alg:creation-lookup-table} simply stored unique grid references in $\tau_a(r)$. Thus, clustering algorithms were not needed.  Moreover, filtering condition ($d(c_i,c_j) < T$) in line \ref{alg:dist} of Algorithm \ref{alg:lateration-lookup-table} was replaced with a Kronecker delta ($\delta_{c_i,c_j} \stackrel{?}{=} 1$)\footnote{defined by Kronecker delta ($\delta_{a,b} = 1 \iff a=b$, otherwise $0$)}, which is more quickly evaluated.

Now we can compare the error obtained by fingerprinting and lookup lateration methods under two assumptions: (1) completeness and (2) unambiguousness. By completeness, we assume that there are enough data observed in each grid. This ensures that error is not induced by data sparsity or poor design. By unambiguousness we assume that the grid size is large enough so that identical observations cannot be found in different grids. This assumption ensures that any observation can be determined unambiguously. Now we can claim the following:

\begin{theorem}
Let $\Expect [\E_{LL} ]$ and $\Expect [\E_{RFP}]$ denote the expected error obtained with LL and RF, respectively. Under the conditions mentioned above, we have
$$\Expect[\E_{LL}] \le \Expect [\E_{RFP}].$$
\label{theorem1}
\end{theorem}

\begin{proof}
	
First, let us consider the fingerprinting method. Let $\B\subset R^{|A|}$ be the measurement space,  $\G$ be the grid system, $P_{xy}$ be the density distribution of measurements belonging to grid $xy\in \G$, and $\mu_{xy}$ be the mean vector of $P_{xy}$. The distance function $d$ in RFP localization implicitly specifies a Voronoi partition of $\B$:  $\{\Region_{xy}\}_{xy\in \G}$, where 
\begin{equation}
\Region_{xy} = \{r\in \B \mid d(r,\mu_{xy}) \le d(r,\mu_{uv})\  \forall uv \ne xy \in \G \}.
\label{eq:region}
\end{equation}
Second, let $M_{RSS}\in \B$ denote a query measurement observed at position $M_{xy}$, which is in grid $xy\in\G$.

If $M_{RSS} \in \Region_{xy}$ and $M_{RSS} \notin \Region_{uv}$ for any other $uv\ne xy\in\G$, then RFP localizes $M_{RSS}$ in the center of grid $xy$. The localization error is $\epsilon_1=d(M_{xy},c_{xy})$, that is, the distance between $M_{xy}$ and the center of the grid $c_{xy}$. Let $\Expect_{xy}[\epsilon_1]$ denote the expected error of this type in grid $xy$. Let $\E_1$ denote this type of error.

If $M_{RSS} \notin \Region_{xy}$, then $\exists uv\in \G$ in which RFP localizes $M_{RSS}$. The error is $\epsilon_2=d(M_{xy},c_{uv})$. Let $\Expect_{xy}[\epsilon_2]$ denote the expected error of this type with respect to grid $xy$. Let $\E_2$ denote this type of error. Note that error $\epsilon_2$ is always greater than error $\epsilon_1$, which implies that
\begin{equation}
\Expect_{xy}[\epsilon_1] < \Expect_{xy}[\epsilon_2]. 
\label{eq:expected_error}
\end{equation}

If $M_{RSS} \in \Region_{xy}$ and $M_{RSS} \in \Region_{uv}$ for $uv\ne xy\in\G$, then that means either $M_{RSS}$ is on the border of two regions $\Region_{xy}$ and $\Region_{uv}$ or $\mu_{xy}=\mu_{uv}$. In the first case, RFP has to choose between the two grids. Let us give some advantage to RFP and assume for the sake of simplicity, it always chooses the correct grid somehow. In the second case, we have $\Region_{xy} = \Region_{uv}$. In our opinion this case happens rarely in practice, so we omit this type of error, and we assume that $\mu_{ab} \ne \mu_{cd}$ for any $ab,cd\in \G$ in the rest of this proof.

The probability that RFP localizes $M_{RSS}$ in its correct grid is

\begin{equation}
p_{xy}(\E_1)= \int_{\Region_{xy}} P_{xy}(r_{a_1}, r_{a_2}, \dots, r_{a_{|A|}})\, dr_{a_1}\dots dr_{a_{|A|}}, 
\end{equation}
and it is indicated by the white area under $P_{xy}$ in Figure \ref{fig:fingerprinting_error}. The probability that RFP localizes $M_{RSS}$ in a grid incorrectly is
\begin{equation}
\begin{aligned}
p_{xy}(\E_2)=& \int_{\B\setminus \Region_{xy}} P_{xy}(r_{a_1}, \dots, r_{a_{|A|}})\, dr_{a_1}\dots dr_{a_{|A|}}\\
 =& 1-p_{xy}(\E_1). 
\end{aligned}
\end{equation}
and it is indicated by the gray area  under $P_{xy}$ in Figure \ref{fig:fingerprinting_error}.

Last, the total expected error for the fingerprinting method can be summarized as follows:
\begin{equation}
\Expect [\E_{RFP}] = \sum_{xy\in\G} p(xy) \left[p_{xy}(\E_1)\Expect_{xy}[\epsilon_1] + p_{xy}(\E_2)\Expect_{xy}[\epsilon_2]\right],
\end{equation}
where $p(xy)$ denotes a priori probability of receiving a measurement from grid $xy$.

Now we examine LL, and let us consider $M_{RSS}=[r_{a_1}, r_{a_2}, ...., r_{a_{|A|}}]$ measured in grid $xy$. Because of the first and second assumptions, we have  $xy\in\tau_{a_i}(r_{a_i})$ and $\bigcap_{i}\tau_{a_i}(r_{a_i}) = \{xy\}$, respectively. This means when Algorithm 2 terminates in line \ref{alg:term}, the set $C^k$ unambiguously contains only one candidate grid, where the measurement was observed. 
Hence, the total expected error for lookup lateration can be summarized as follows:
\begin{equation}
\Expect [\E_{LL}] = \sum_{xy\in\G} p(xy) \Expect[\E_1],
\end{equation}
Following from Eq. \ref{eq:expected_error} we obtain
\begin{equation}
\Expect [\E_{LL}] \le \Expect[ \E_{RFP}],
\end{equation} 
which proves our claim.

\end{proof}

To illustrate the proof of Theorem \ref{theorem1} , let us consider a scenario in which there is one antenna tower $A$ and two nearby grids $G_{xy}, G_{uv}$. Furthermore, let us assume that $A$ receives $-53, -55, -57, -59$, and $-61$ dBm from $G_{xy}$ and $-58, -60, -62, -64$, and $-66$ dBm from $G_{uv}$. The mean values are $\mu_{xy}=-57$ and $\mu_{uv}=-62$ and $\Region_{xy}=\{x \mid x\ge -59.5\}$ and $\Region_{uv}=\{x\mid x \le -59.5\}$. RFP locates measurement $M=-61$ dBm in grid $G_{uv}$ incorrectly because $M\in\Region_{uv}$, that is, $M$ closer to $\mu_{uv}$ than to $\mu_{xy}$. However, LL constructs a table $\tau_A$ in which $\tau_A (-66)=\tau_A (-64)=\tau_A (-62)=\tau_A (-60)=\tau_A (-58)=\{G_{uv}\}$ and $\tau_A (-61)=\tau_A (-59)=\tau_A (-57)=\tau_A (-55)=\tau_A (-53)=\{G_{xy}\}$. Therefore, LL will find $G_{xy} \in \tau_A (-61)$  and localize $M$ in it correctly.

\begin{figure}[hbtp]
	\centering
	\small
	\includegraphics[trim=4cm 8cm 8cm 1.5cm,clip=true, width=70mm]{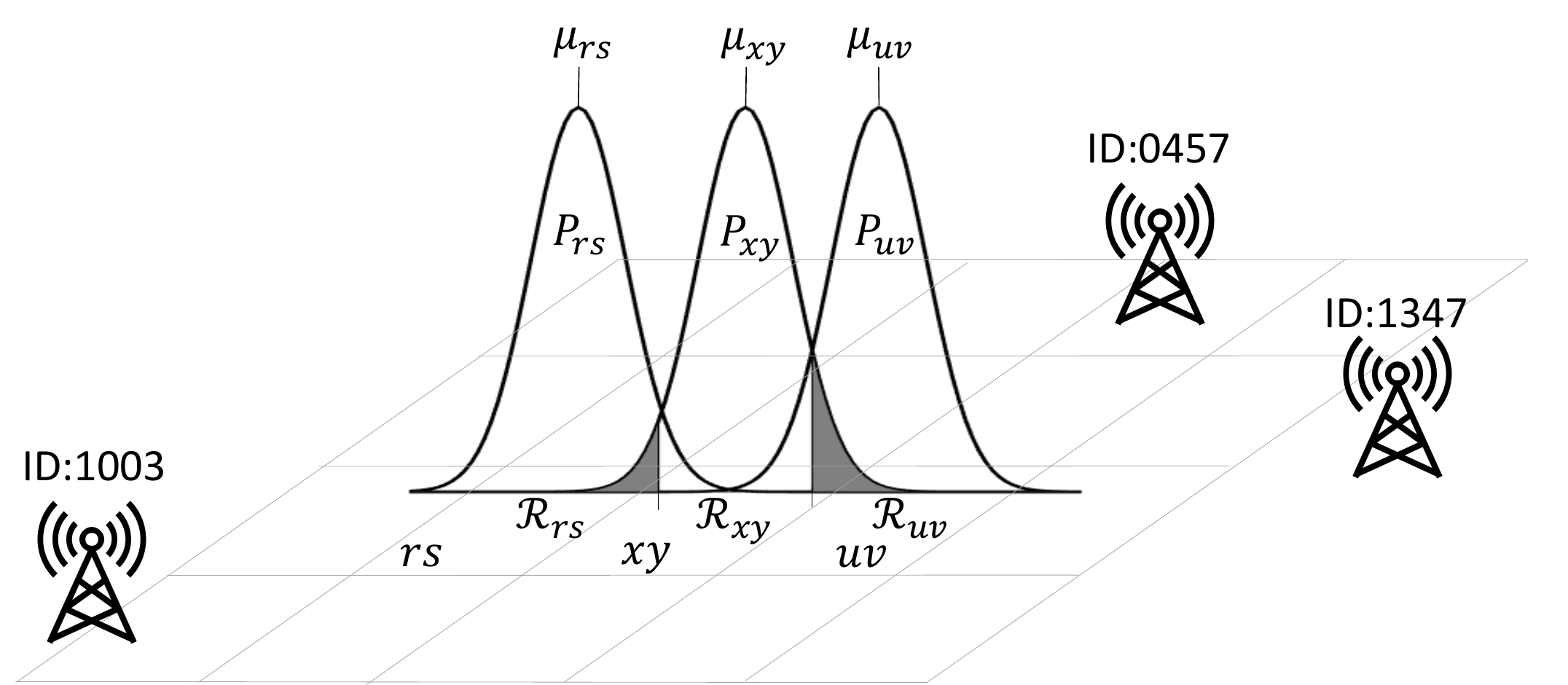}
	\caption{Distribution of measurements for three grids $rs, xy,$ and $uv$ in fingerprinting scenario. Decision boundaries between $\Region_{rs},\Region_{xy}$, and $\Region_{uv}$ are halfway between distribution means.  Consider a query in grid $xy$. Measurements belonging to region $\Region_{xy}$ are correctly localized in grid $xy$, and the white area in the distribution $P_{xy}$ denotes $p_{xy}(\E_1)$. However, measurements belonging to other regions are misplaced in other grids, and the gray area denotes $p_{xy}(\E_2)$. Note that distributions are multivariate and can be different from Gaussian.}
	\label{fig:fingerprinting_error}
\end{figure}

One may argue that Theorem \ref{theorem1} is based on two conditions that are hard to ensure in practice. We may agree, but in our opinion, Theorem \ref{theorem1} shows the conceptual limitation of the fingerprinting method. However, if we easy up on the conditions and if the unambiguousness is not required, then we claim 
\begin{theorem}
	Fingerprinting method is suboptimal
	\label{theorem2}
\end{theorem}
\noindent in the sense of the Neyman-Pearson criterion.

\begin{proof}
	Localization using grid systems can be considered a classification problem in which grids are considered classes, and a localization method has to decide whether to classify a query $M_{RSS}$ in a given class $xy$. 
	Let $P^+(M_{RSS}) = P_{xy}(M_{RSS})$ and $P^-(M_{RSS})=\frac{1}{Z}\sum_{uv\ne xy}P_{uv}(M_{RSS})$ be likelihood functions, where $Z$ is an appropriate normalization factor.  According to the Neyman-Pearson Lemma (NPL), the highest sensitivity can be achieved with the following likelihood-ratio rule:
	\begin{equation}
	\text{Localize } M_{RSS} \text{ in grid } xy \text{ if } \quad \frac{P^+(M_{RSS})}{P^-(M_{RSS})} \ge \eta,
	\end{equation} 
	where $\eta$ is a trade-off parameter among false positive, false negative error, and statistical power. Let $\eta=1$ for sake of simplicity. The region $\Region^+=\{M_{RSS}\mid P^+(M_{RSS})\ge P^-(M_{RSS}) \}$ in which a query $M_{RSS}\in \Region^+$ is localized in grid $xy$ and defined by the likelihood ratio can be noncontinuous region. However, Figure \ref{fig:fp_vs_np} illustrates that the region $\Region_{xy}$ defined by RFP in eq. \ref{eq:region} is always continuous and different from $\Region^+$. Therefore, RFP yields less or equal sensitivity than it could be achieved with NPL, where equality holds iff all distribution belonging to grids are symmetric and unimodal. 
\end{proof}
\begin{figure}[btp]
	\centering
	\small
	\includegraphics[trim=0cm 4cm 0cm 7cm,clip=true, scale=0.3]{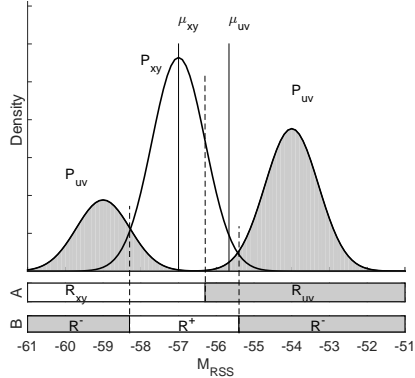}
	\caption{Decision boundaries in case of multi-modal density distributions. $P_{uv}$ is a multi-modal distribution that might be obtained from a partially shadowed area. Measurements obtained from the open field might concentrate on one mode while measurements obtained from behind a building or from an underpass might concentrate around the other mode of the distribution.
		Then $\mu_{xy},\mu_{uv}$ denote the mean vector of the two distributions, respectively. The decision boundary defined with RFP is located between $\mu_{xy},\mu_{uv}$ and the corresponding regions $\Region_{xy},\Region_{xy}$ denoted on A. For instance, query measurement $M_{RSS} = -59$ will be localized in grid $xy$ with RFP; however, it is more likely to be in grid $uv$. Decision boundaries defined by the NPL are shown on B.}
	\label{fig:fp_vs_np}
\end{figure}

In other words, the weakness of RFP arises from the fact that RFP does not take into account the full distribution of measurements. It uses only the mean vectors of the distributions, while NPL utilizes the whole distribution. 

On the other hand, LL will identify all candidate grids in which the query can be found. Therefore, the set of candidate grids $C^{(k)}$ will contain the correct grid as well. If LL was programmed to localize a query by the center of the grid $\tilde{xy}$ for which $\tilde{xy}=\text{argmax}_{xy\in C^{(k)}}\{P_{xy}(M_{RSS})\}$, then LL would be optimal in the sense of NPL. However, we decided to report the mean of the candidate grids. In this case, $\E_2$ happens, but we hope averaging the candidate locations would mitigate the amount of error $\epsilon_2$. This also does not require us to store or model full measurement distributions.

\section{Results and Discussion}

We have carried out our experiments on the Urban Hannover Scenario dataset \cite{Rose2013}. This dataset contains approximately 22 Gb of RSS measurements simulated in Downtown Hannover, along with a reference $x$-$y$ location. The reference point (0-0) for the coordinate system is the lower left corner of the scenario. The data is the result of a prediction with a calibrated ray tracer using 2.5D building information. For each mobile phone, the 20 strongest RSS measurements are provided. For further details, we refer to \cite{Rose2013}.  Data was split into training and test sets randomly, and experiments were repeated ten times. Results were averaged. The variance in results was very small because of the dataset's huge size; therefore, standard deviation is not shown for the sake of simplicity. Algorithms were implemented in Python programming language and executed on a PC equipped with a 3.4 GHz CPU and a 16 GB RAM.

\begin{figure}[hbtp]
	\centering
	\small
	\begin{tabular}{ccc}
		\includegraphics[trim=1cm 0cm 1cm 0cm,clip=true, width=43mm]{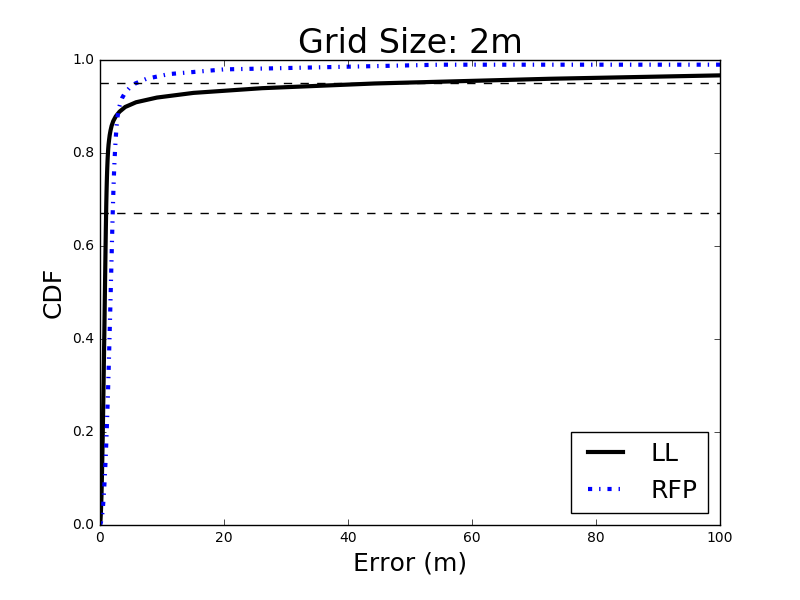}&
		\includegraphics[trim=1cm 0cm 1cm 0cm,clip=true, width=43mm]{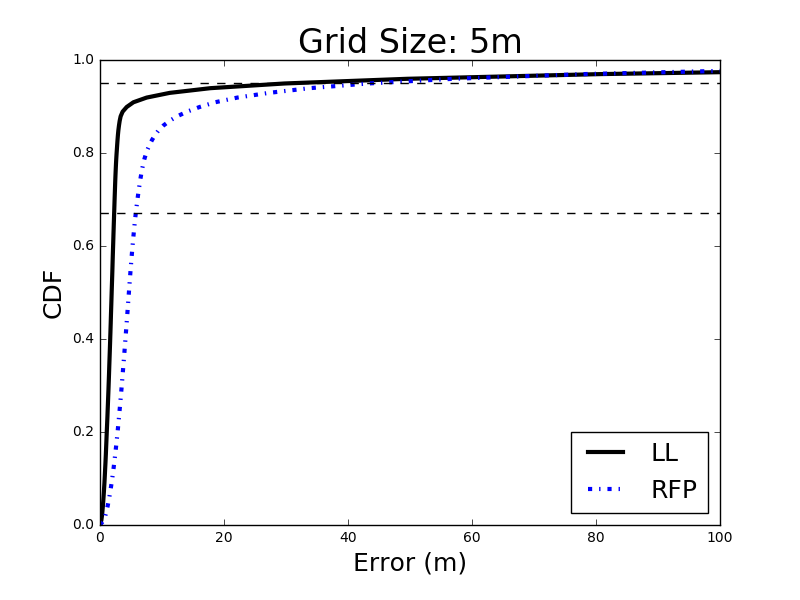}&
		\includegraphics[trim=1cm 0cm 1cm 0cm,clip=true, width=43mm]{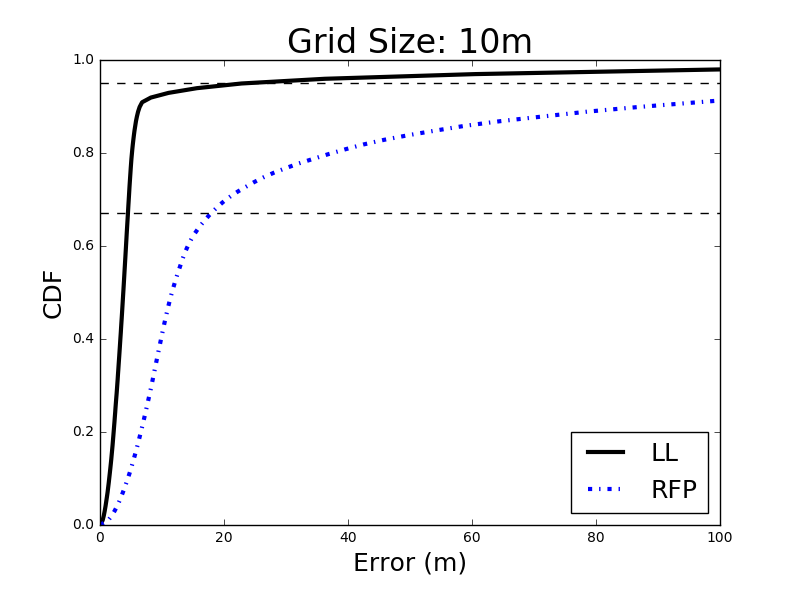} \\
		(A) & (B) & (C)\\
		\includegraphics[trim=1cm 0cm 1cm 0cm,clip=true, width=43mm]{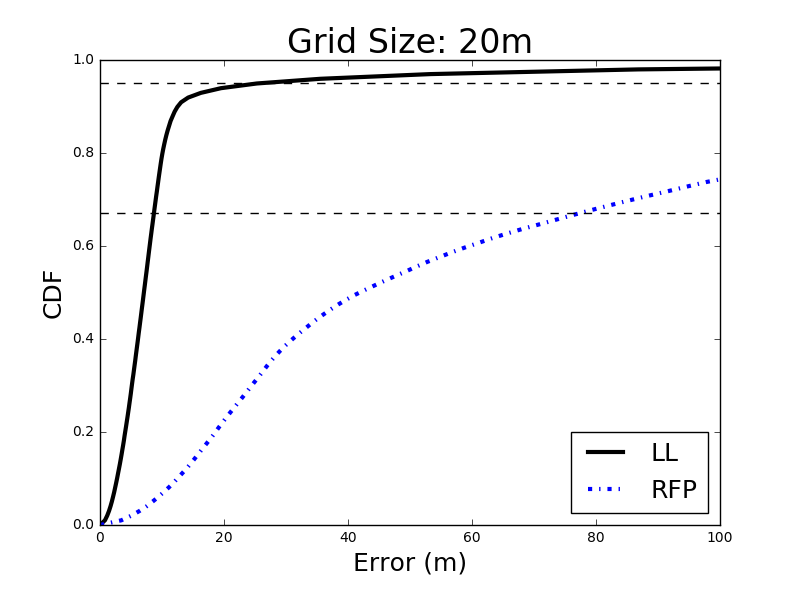} &
		\includegraphics[trim=1cm 0cm 1cm 0cm,clip=true, width=43mm]{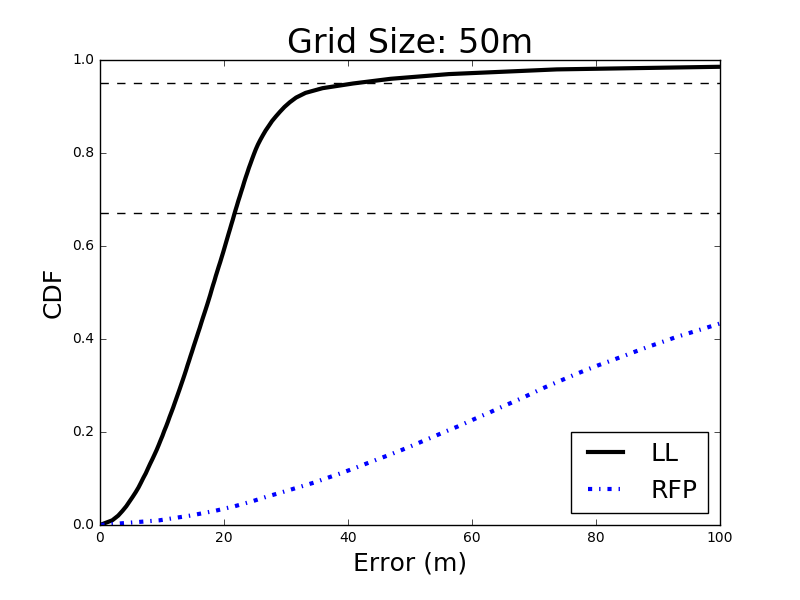}&
		\includegraphics[trim=1cm 6cm 2cm 6cm,clip=true, width=43mm, height=36mm]{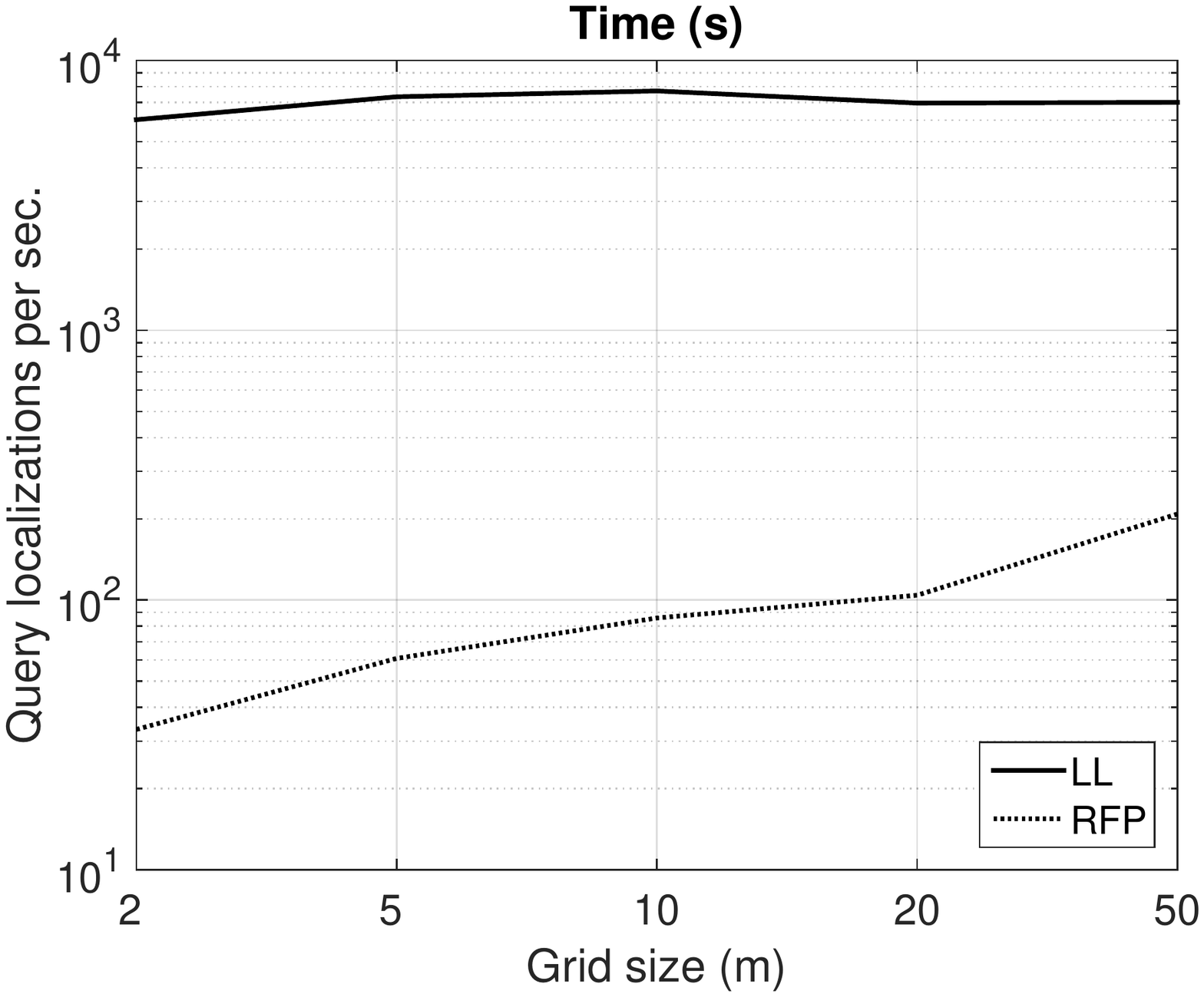}\\
		(D) &(E)&	(F)\\
	\end{tabular}	
	\caption{Performance of RFP and LL using various grid sizes. The training set size was 10 \%. }
	\label{fig:main_results}
\end{figure}

First, we investigated the localization accuracy of RFP and LL methods using grid sizes 5, 10, 20, and 50 meters. The full dataset was split to 10\% training and 90\% test data. The results shown in Figure \ref{fig:main_results} clearly tell us LL and RFP perform nearly the same when the grid size is small (Figures \ref{fig:main_results}A$-$B). The performance of LL remains roughly the same as the size of the grid grows (Figures \ref{fig:main_results}C$-$D), while the performance of RFP decreases quickly. In our opinion the dramatic drop in the performance of RFP is in accordance with Theorems \ref{theorem1} and \ref{theorem2}. Large grids cover a large area containing various obstacles on the ray propagation path, which result in wide, multimodal measurement distributions w.r.t. grids. This is not taken into account by RFP. Next, we analyzed the speed of these methods, and the results are summarized in Figure \ref{fig:main_results}F. LL seems to be extremely fast, around two magnitudes faster compared to the RFP method. This improvement stems from the fact that LL processes only a few lookup tables related to a query. On the other hand, RFP iteratively processes all grids while seeking  the most similar RSS vector pattern. It is also surprising, the execution time of LL does not depend on the grid size, while RFP quadratically becomes slower as the grid size decreases.

One of the main drawbacks of any RFP-based method is that they require a well-designed reference database because a query cannot be localized in its correct grid if it was not covered during site surveying. This is the case for LL as well because it cannot localize measurements in grids that are not stored in lookup tables. Thus, we investigated how these methods perform when the localization phase is not preceded by a proper site surveying. We calculated a site coverage defined as the ratio of grids that contain training data. For instance, 35\% of coverage means that the training data belong to 35\% of the total grids, and all queries would be localized in one of these grids. Therefore, 65\% of the grids contain queries that would be localized in wrong grids. We note that the larger the coverage, the more complete the reference dataset and the smaller error we expect. The coverage is driven by two factors: (i) grid size and (ii) the size of the training data. First, we took 5\% of the training data and calculated the coverage and the localization error at 67\% and 95\% using various grid sizes. The results shown in Table \ref{tab::RES_COV} tell us a larger grid size results in larger coverage, and both methods yield a larger localization error at 67\%. However, if we take a closer look and calculate a relative error (RE) as the ratio of the error and the grid size, we can observe opposite tendencies. The RE obtained with LL decreases as the grid size grows, and we think this is the result of increased coverage. However, the RE obtained with RFP increases in spite of increased coverage, and we explain this by the arguments in Theorems \ref{theorem1} and \ref{theorem2}. On the other hand, the RE and the overall error show opposite tendencies at 95\%. LL decreases while RFP increases the overall error and RE as the sizes of grids and coverage grow. 



\begin{table}
	\caption{Localization Error (m) and Coverage with various grid sizes}
	\begin{center}
		\footnotesize
		\begin{tabular}{cccccccccccc}
			\hline
			\multicolumn{2}{r}{Grid size (m)} &\multicolumn{2}{c}{1}&\multicolumn{2}{c}{5}&\multicolumn{2}{c}{10}&\multicolumn{2}{c}{20}&\multicolumn{2}{c}{50}\\
			\hline
			\multicolumn{2}{r}{Cov. (\%)} &\multicolumn{2}{c}{26.38}&\multicolumn{2}{c}{54.42}&\multicolumn{2}{c}{66.74}&\multicolumn{2}{c}{71.78}&\multicolumn{2}{c}{87.20}\\
			\hline
			& & Error &RE$^1$& Error &RE& Error &RE& Error &RE& Error &RE\\
			\hline	
			\multirow{2}{*}{67\%}&LL&0.77&\em 0.77&2.44&\em 0.49&4.74&\em 0.47&9.11&\em 0.46&22.24&\em 0.44 \\
			&RFP&1.09&\em 1.09&5.76&\em 1.15&16.36&\em 1.64&72.50&\em 3.63&157.17&\em 3.14\\
			\hline	
			\multirow{2}{*}{95\%}&LL&125.98&\em 125.98&89.51&\em 17.9&69.37&\em 6.94&55.66&\em 2.78&49.88&\em 0.998\\
			&RFP&6.32&\em 6.32&45.05&\em 9.01&138.62&\em 13.68&242.19&\em 12.11&343.94&\em 6.88\\
			\hline	
			\label{tab::RES_COV}
		\end{tabular} 
	\end{center}
	\vspace{-6mm}
	{\footnotesize The training set size was 5\%. $^1$Relative error (RE) is defined as the ratio of error and grid size.} 
\end{table}

Next, we fixed the grid size to 5 m, and we varied the training set size from 20\% to 1\%. The results are shown in Figure \ref{fig:trainsize_results}. It can be seen that the performance of RFP seems to be unaffected by training set sizes; however, the performance obtained with LL decreases as coverage shrinks.


\begin{figure}[hbtp]
	\centering
	\small
	\begin{tabular}{cc}
		\includegraphics[trim=1cm 0cm 1cm 0cm,clip=true, width=60mm]{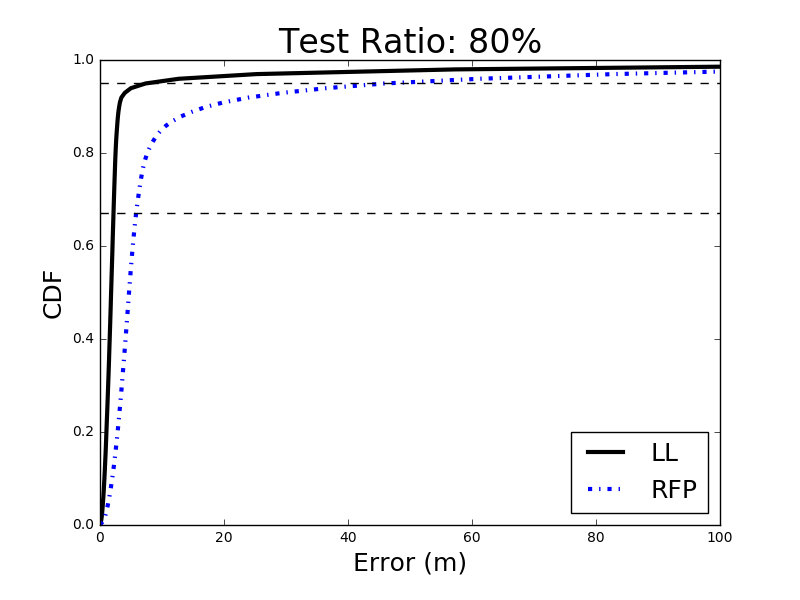}&
		\includegraphics[trim=1cm 0cm 1cm 0cm,clip=true, width=60mm]{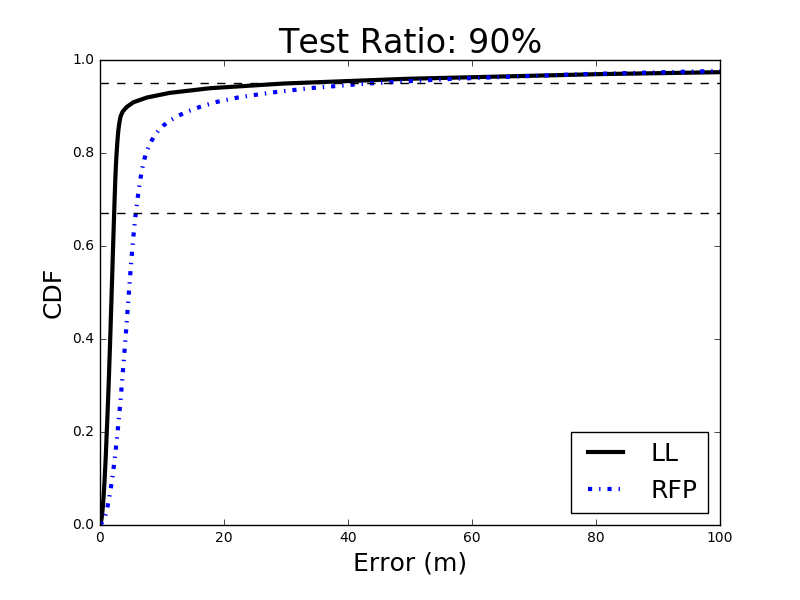}\\
		Coverage: 80.43\%&Coverage: 68.15\%\\
		\includegraphics[trim=1cm 0cm 1cm 0cm,clip=true, width=60mm]{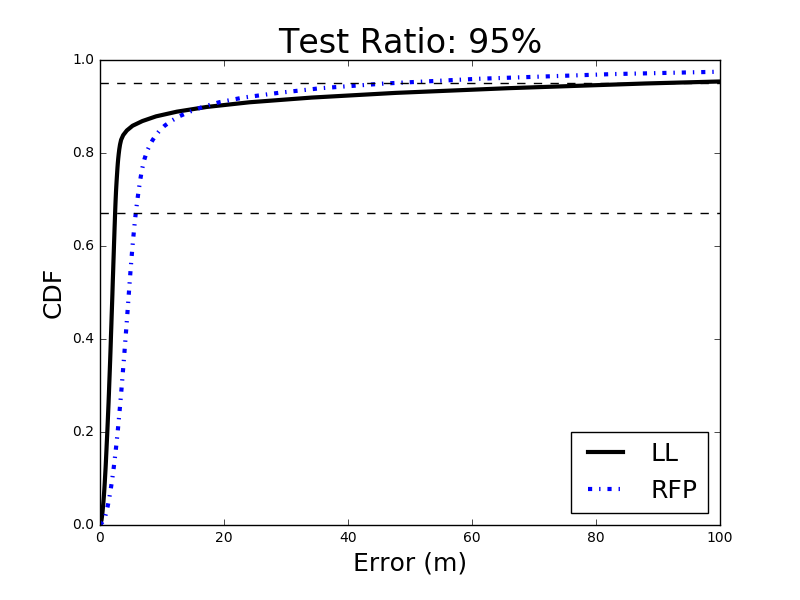}&
		\includegraphics[trim=1cm 0cm 1cm 0cm,clip=true, width=60mm]{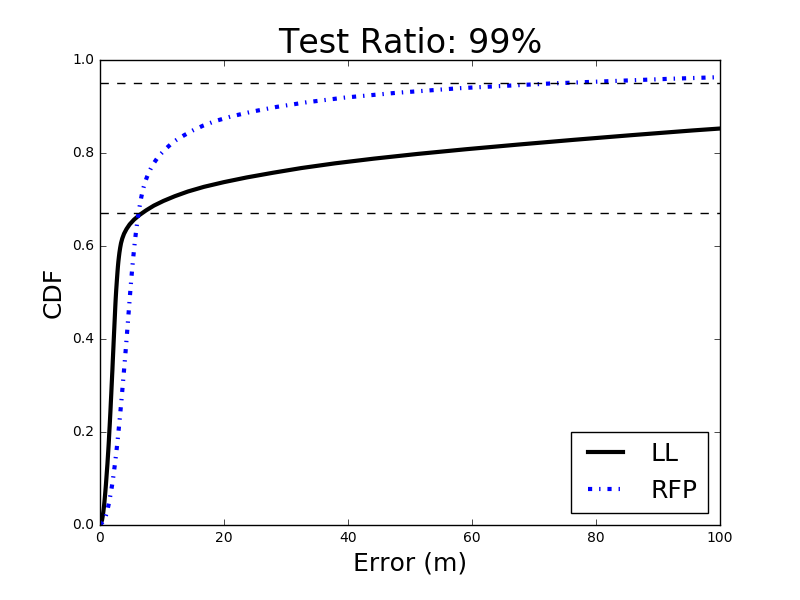}\\
		Coverage: 54.42\%&Coverage: 26.32\%\\
	\end{tabular}	
	\caption{Performance of RFP and LL using various training data sizes. Grid size was 5 m.}
	\label{fig:trainsize_results}
\end{figure}

Now we study the bin size $s$ used in LL and its effect on performance. We run LL using various bin sizes $s=1,2,3,5,10$ and training set sizes $1\%,5\%,10\%$, and $20\%$. The grids' size was fixed to 1 m, and the results are shown in Table \ref{tab::RES_BINNING}. We observed that in extreme cases when the coverage is very small, doubling bin size to $s=2$ can halve the localization error at 95\%. For instance, when coverage is only 26.57\%, the error at 95\% reduces from 125.95 m to 45.22 m. This could be further reduced to 21.70 m by using larger bin size $s=5$. Surprisingly, the bin size does not seem to have a strong impact on error at 67\%. This means that increasing bin size can compensate for improper site surveying in localization accuracy.

\section{Conclusions}
In this article, we have presented a new method called lookup lateration for localization problems in densely populated urban areas using received signal strength (RSS) data. Our method combines the advantages from both triangular lateration and fingerprinting. 

Akin to fingerprinting, lookup lateration utilizes a database of localized RSS as reference. This ensures high accuracy in urban areas because it implicitly encodes NLOS information in the reference dataset. However, the main difference from fingerprinting is that fingerprinting searches one big dataset for the most similar reference point in a nearest neighbor fashion. Contrary to that, lookup lateration stores reference locations for different RSS values and antenna towers in lookup tables separately. Then each lookup table can be considered a nonlinear RSS-to-distance mapping viewed from a given antenna, and mobile device's localization can be carried out by simply identifying common reference points in corresponding lookup tables. This process resembles triangular lateration apart from the fact that a non-linear RSS-to-distance mapping is utilized without involving any optimization and local minima problems. Furthermore, the main benefit of the decentralization of reference datasets is that the localization process does not need to search the whole dataset, just the few lookup tables w.r.t. observed RSS from corresponding antennas. This results in a speedup by two magnitudes independently from the size of grids and also ensures that lookup lateration is scalable and can be performed in distributed systems easily, akin to triangular lateration.

Lookup lateration outperforms triangular lateration and fingerprinting methods in localization accuracy as well. Triangular lateration uses monotone transformation to calculate distance from RSS values, and it does not take into account NLOS objects and buildings.  Hence, its performance in urban areas is always limited. In Theorems \ref{theorem1} and \ref{theorem2} we showed a conceptual limitation of the fingerprinting method when it is used with grid systems to reduce redundancy and search time. Briefly, the problem is that RSS observations are aggregated in the same grid position, and their variances are not taken into account. As a consequence, fingerprinting is prone to annotate observations with incorrect grid positions. This situation cannot happen with lookup lateration because it stores all grid positions for any observed RSS values separately. This fact makes LL very robust to larger grid sizes.

At last, we need to mention a potential drawback of our method, which is sensitivity to data sparsity. If we are given extremely limited data from site surveying and lookup tables do not contain enough candidate locations from every tower w.r.t. RRSs, then LL can be outperformed by RFP. However, we have shown that the data sparsity can be compensated for by increasing the RSS bin size. This can improve localization accuracy efficiently. 

Finally, it is very easy to migrate from RFP to LL. Since LL does not need any additional data compared to RFP, lookup tables can be constructed from the site surveying data of RFP. Lookup tables can also be an intelligent reorganization of fingerprinting data, which preserves measurement variance implicitly.

Our method is very simple, easy to implement in distributed systems, and inexpensive to maintain, which, we hope, can make it appealing for large-scale industrial applications.

\bibliographystyle{plain}
\bibliography{bibliography}

\end{document}